\documentclass[showpacs,preprintnumbers,aps,prl,letterpaper,superscriptaddress,nofootinbib,tightenlines,floats,floatfix]{revtex4}
\usepackage{graphicx}
\usepackage{bm}
\usepackage{latexsym}
\usepackage{amsmath,amssymb}
\usepackage{amsfonts}
\usepackage[draft=false]{hyperref}
\usepackage[latin1]{inputenc}
\usepackage{amsthm}
\usepackage{mathrsfs}
\usepackage{ifpdf}
\usepackage{color}
\usepackage{epstopdf}
\usepackage[skip=2pt,font=scriptsize]{caption}
\usepackage{natbib}
\usepackage{blindtext}

\theoremstyle{definition}
\newtheorem{definition}{Definition}

\newtheorem{theorem}{Theorem}
\newtheorem*{corollary}{Corollary}
\newtheorem{lemma}{Lemma}

\theoremstyle{remark}
\newtheorem*{remark}{Case $K=-\alpha=\beta$}

\newtheorem*{remarkk}{Remark}

\begin{document}
\title{Geodesics Near a Curvature Singularity of Stationary and Axially Symmetric Space-Times}

\author{Juan Carlos Del \'Aguila}
\email{jdelaguila@fis.cinvestav.mx}
\affiliation{Departamento de F\'isica, Centro de Investigaci\'on y de
  Estudios Avanzados del IPN, A.P. 14-740, 07000 Ciudad de M\'exico, M\'exico.}
\author{Tonatiuh Matos}
\email{tmatos@fis.cinvestav.mx}
 \altaffiliation{Part of the Instituto Avanzado de Cosmolog\'ia (IAC)
  collaboration http://www.iac.edu.mx/}
\affiliation{Departamento de F\'isica, Centro de Investigaci\'on y de
  Estudios Avanzados del IPN, A.P. 14-740, 07000 Ciudad de M\'exico,
  M\'exico.}

\begin{abstract}
In this work we study the local behavior of geodesics in the neighborhood of a curvature singularity contained in stationary and axially symmetric space-times. Apart from these properties, the metrics we shall focus on will also be required to admit a quadratic first integral for their geodesics. In particular, we search for the conditions on the geometry of the space-time for which null and time-like geodesics can reach the singularity. These conditions are determined by the equations of motion of a freely-falling particle. We also analyze the possible existence of geodesics that do not become incomplete when encountering the singularity in their path. The results are stated as criteria that depend on the inverse metric tensor along with conserved quantities such as energy and angular momentum. As an example, the derived criteria are applied to the Pleba\'nski-Demia\'nski class of space-times. Lastly, we propose a line element that describes a wormhole whose curvature singularities are, according to our results, inaccessible to causal geodesics. 
\end{abstract}

\date{Received: date / Accepted: date}

\maketitle

\date{\today}


\maketitle

\section{I. Introduction}
Space-time singularities are, to this day, one of the aspects of General Relativity (GR) which still hold several unanswered questions. Difficulties arise even from the supposedly simple task of exactly defining what a singularity is and whether a space-time is singular or not. Over the years, an useful approach to identifying singularities has been that of examining the curves of the space-time manifold, specifically geodesics, and its affine completeness, i.e. the ability to extend such curves to, either past or future, arbitrary values of its affine parameter. Based on this global property of the manifold, there is general agreement on considering that a geodesically incomplete space-time is, by all means, singular \cite{wald}. A very comprehensive review on this subject can be found in \cite{Senovilla1,Senovilla2}.  

Many widely known exact solutions to the Einstein equations contain a singularity, for instance, those who describe black holes. In fact, it is thought that once gravitational collapse takes place, the formation of a space-time singularity is unavoidable \cite{penrose}. In the case of black holes, its singularities are commonly associated to the divergence of curvature scalars in the space-time metric. This has led to the notion of the necessity of a quantum theory of gravity to adequately describe regions where said scalars approach to the Planck scale. Nevertheless, these particular singularities are physically accepted since they are hidden behind event horizons and hence, causally disconnected from the outer region of the black hole. On the other hand, space-times with unbounded curvature scalars, and that are not equipped with an event horizon (the so-called naked singularities), are often dismissed as non-physical and thus, considered as pathological. As examples of this type of space-times, the extremal Kerr metric ($a>m$) and the axially symmetric ring wormholes can be mentioned \cite{ringWH}.

While in black hole solutions the concepts of singularity and unbounded curvature scalars seem to be deeply related, recent papers have provided examples of causal geodesically complete wormhole space-times despite the presence of diverging curvature \cite{Olmo1}. Further study of such metrics has shown that even the curves of observers with bounded acceleration are complete \cite{Olmo2}, therefore extending the argumentation of the regularity of those space-times. Examples of the opposite, that is, bounded curvature with geodesic incompleteness, are also known. One of them is given by Misner in \cite{misner} and consists of a metric with the properties of a Taub-NUT vacuum space-time. Hence, it is clear that in regards to geodesic incompleteness and the divergence of curvature scalars, one does not imply the other. An analysis of these two features, as well as unbounded energy density, is discussed in \cite{Olmo3} for the case of space-times in a quadratic $f(R)$ gravity theory.

The wormhole metrics cited before contribute to consider cases in which the divergence of curvature, referred to in this paper as ``curvature singularity", is not to be seen as singular or badly behaved, so long as it does not induce pathological effects on the curves of physical observers. In fact, on the matter of general singular space-times, Clarke had previously introduced the modern perspective that singularities are not to be generally seen as obstacles to extend geodesics but as obstacles to evolve test fields \cite{Clarke}. For this purpose, a new concept analogous to global hyperbolicity is proposed, ultimately heading to the possibility that some apparent singularities do not represent a breakdown of cosmic censorship \cite{cc}. All of these results can lead to a more open acceptance of certain singularities, or at least, to not discard so easily space-times which contain them. 

Following on the idea that unbounded curvature does not imply geodesic incompleteness, and vice-versa, in this work we try to establish the necessary and sufficient conditions for which causal geodesics in a space-time containing diverging scalars can reach the curvature singularity within it. We also consider the possibility that if, unavoidably, said singularity is met by a given curve, completeness could still be possible. In this paper we present two theorems that serve mainly as criteria for the occurrence of these two particular behaviors. Here, we will consider four-dimensional, stationary and axially symmetric space-times, with some additional requirements that will be explicitly mentioned in section II. In that section some concepts and results that shall help us throughout the analysis will also be discussed. Afterward, we state the theorems and give their proofs. In section III we apply them to a physically relevant class of space-times, that of Pleba\'nski-Demia\'nski, and obtain some well-known results. Finally, based on the first derived theorem, we construct a metric in which causal geodesics are unable to touch its curvature singularity. Proving thus, with the aid of some additional arguments, its causal geodesic completeness despite the presence of unbounded curvature.  

\section{II. Curvature Singularities and Geodesics}

We begin by establishing various concepts that will allow us to deal with singularities throughout the paper. As mentioned in the Introduction we shall focus on the so-called curvature singularities. Let $u^\mu=\{u^0=t,u^1,u^2,u^3\}$ be Cartesian (or at least ``Cartesian-like'') coordinates on a (pseudo)-Riemannian manifold $M$. We will say a space-time $(M, g_{\mu\nu})$ contains a curvature singularity if any of its curvature scalars $R_X$ diverge at some coordinate values $u_0^\mu$. The scalars $R_X$ may be constructed from index contractions or from polynomial expressions of the Riemann or Ricci tensors. In the case of a $n$-dimensional manifold, the curvature scalars can be considered as a map $R_X:\mathbb{R}^n\rightarrow\mathbb{R}$. With this in mind, we make the following definitions for a 4-dimensional space-time.

\theoremstyle{definition}
\begin{definition}
Let $R_X:\mathbb{R}^4\rightarrow\mathbb{R}$ denote a curvature scalar of a given metric $g_{\mu\nu}$ that contains a curvature singularity labeled as $\sigma$. The singular curvature set is defined as $\sigma_X=\left\{(u^0,u^1,u^2,u^3)\in\mathbb{R}^4\mid 1/R_X(u^\mu)=0\right\}$.

\end{definition}

\begin{remarkk}
This definition is suitable to identify curvature singularities that arise as a consequence of a vanishing denominator of $R_X$, this is the case in some situations in GR. Nevertheless, it does not absolutely defines all possible curvature singularities that can exist, as some may escape this definition\footnote{For instance, consider a scalar that goes as $R_X\sim e^{u^1}$. In this case $1/R_X$ is never exactly zero, and so $\sigma_X=\{\varnothing\}$, despite the fact that the curvature is evidently unbounded.}.
\end{remarkk}

The coordinate system $\{u^\mu\}$ used in this past definition is referred to as Cartesian-like in the sense that there exist appropriate limits on the parameters of the space-time for which $\{u^1,u^2,u^3\}$ become regular Cartesian coordinates in Euclidean 3-space.

\theoremstyle{definition}
\begin{definition}
The singular curvature set will be said to be spatially compact/bounded/open, etc., if its subset $(u^1,u^2,u^3)\in\mathbb{R}^3$ is compact/bounded/open, etc.

\end{definition}

The singular curvature set of an asymptotically flat space-time is either empty or spatially bounded. In this paper we will treat only singularities whose $\sigma_X$ is spatially compact. 

Since by definition, a space-time is constituted of only regular points, a singularity does not properly belong to it. This implies that a neighborhood of the singular points cannot be defined in the usual topological sense. However, using an auxiliary manifold $\tilde{M}$, the neighborhood of a curvature singularity may be ultimately defined.

\theoremstyle{definition}
\begin{definition}
Let ($M$, $g_{\mu\nu}$) be a space-time that contains a curvature singularity $\sigma$. Also, let $\zeta:M\rightarrow\tilde{M}$ be a non-isometric embedding, being $\tilde{M}$ a manifold containing all the points of the set $\sigma_X$ and so, $M\subset\tilde{M}$, i.e. $M$ is a proper subset of $\tilde{M}$. Then, the neighborhood $N$ of the singularity is $N=\tilde{N}\cap M$, where $\tilde{N}$ is a neighborhood of $\sigma_X$ in $\tilde{M}$.

\end{definition}

Note that $\zeta$ must be a non-isometric embedding so that $\sigma_X$ is not singular on $\tilde{M}$. With the neighborhood of the singularity properly defined, we can distinguish between certain types of singularities depending on the nature of space-time events that take place in $N$.

\theoremstyle{definition}
\begin{definition}
A curvature singularity $\sigma$, generated by a scalar $R_X$ that depends on the spatial coordinates $x^i,\ldots,x^j$, will be called time-like if there exists a neighborhood $N$ of $\sigma$, in which the coordinate vector fields $\partial/\partial x^i,\ldots,\partial/\partial x^j$ are everywhere space-like. The neighborhood $N$ will be called a time-like neighborhood of the singularity.

\end{definition}

From this definition, one can see that a particle lying inside a neighborhood of a time-like singularity will not necessarily meet the singularity in the future of its world line.

These concepts shall be later applied to axially symmetric line elements. In this paper we will be interested in four-dimensional space-times $(M,g_{\mu\nu})$ that possess the following set of properties:

\begin{enumerate}
	\item Stationary, axially symmetric and satisfying the circularity condition\footnote{If $w_0$ and $w_3$ are the one-forms associated to the two commuting Killing vector fields $X_0$ and $X_3$, respectively, that exist as consequence of this first property, then the condition $w_0\wedge w_3\wedge dw_0=w_0\wedge w_3\wedge dw_3=0$ shall be called the circularity condition. This is equivalent to the Ricci circularity property $w_0\wedge w_3\wedge R(w_0)=w_0\wedge w_3\wedge R(w_3)=0$, where here $R(w_i)=R_{\mu\nu}X_i^\nu$ is the Ricci form ($i=0,3$) (see sections 8.3.1 and 8.3.2 of \cite{straumann} for reference).}
	\item Its geodesics admit a non-trivial\footnote{By non-trivial we mean a quadratic first integral other than the one yielded by the metric itself, since $g^{\mu\nu}p_{\mu}p_{\nu}$ is constant.} quadratic first integral
	\item Contains a time-like curvature singularity $\sigma$ whose singular curvature set $\sigma_X$ is non-empty.
	\item There exists an unphysical space-time $(\tilde{M},\widetilde{g}_{\mu\nu})$ such that $M\cup\sigma_X\subseteq\tilde{M}$ and $g^{\mu\nu}=\widetilde{g}^{\mu\nu}/\tau$ with $\widetilde{g}^{\mu\nu},\tau\in C^{\infty}$ in a neighborhood $\tilde{N}\subset\tilde{M}$ of $\sigma_X$.
\end{enumerate}

The circularity condition of property 1 holds for a wide class of energy-momentum tensors of great physical interest. For instance, vacuum space-times, Einstein-Maxwell fields, perfect fluid solutions with circular flow, and real scalar fields solutions \cite{stephani}. Therefore, the results here presented are not restricted to specific solutions of the Einstein field equations. We would also like to point out that property 2, at least for null geodesics, is fulfilled for any algebraically special space-time of type D \cite{killing}.

We now develop some auxiliary results regarding the implications of properties 1 to 4, which will be later used in the proof of the main theorem of the paper.

\begin{lemma}
If the geodesics of a four-dimensional, stationary, axially symmetric, and circular space-time $(M,g_{\mu\nu})$ admit a non-trivial quadratic first integral, then there exists a coordinate system $\left\{x^{\mu}\right\}$ in which $$g^{\mu\nu}=\left[L^{\mu\nu}(x^1)+\Theta^{\mu\nu}(x^2)\right]/\left[f(x^1)+h(x^2)\right],$$ with $L=L^{ij}\partial_i\otimes\partial_j+L^{11}\partial_1\otimes\partial_1$ and $\Theta=\Theta^{ij}\partial_i\otimes\partial_j+\Theta^{22}\partial_2\otimes\partial_2$ ($i,j=0,3$).
\end{lemma}

\begin{proof}
It follows from stationarity and axial symmetry that $g_{\mu\nu}$ is characterized by two commuting Killing vector fields $X_0=\partial/\partial t$ and $X_3=\partial/\partial\varphi$, where we have introduced coordinates $x^0=t$ and $x^3=\varphi$, which in flat space-time can be given the physical interpretation of coordinate time and azimuthal angle, respectively. Furthermore, to each corresponding Killing vector field there is an associated momentum $p_i=\partial\mathcal{L}/\partial\dot{x}^i$ that remains constant along geodesic curves\footnote{Throughout the rest of this paper Greek indices will run from 0 to 3 (as used conventionally), while lower-case Latin indices will only take the values of 0 and 3 ($i=0,3$), and upper-case Latin indices will take the values of 1 and 2 ($A=1,2$).}. Here, $\mathcal{L}=g_{\mu\nu}\dot{x}^\mu\dot{x}^\nu/2$ is the Lagrangian of a freely falling particle in the space-time and $\dot{x}^\mu=dx^\mu/d\lambda$ are its coordinate velocities with respect to an affine parameter $\lambda$. The Killing vectors $X_i$ represent isometries of the space-time in the directions $x^i$. Hence, having previously stated the physical meaning of these coordinates, we can relate the momenta $p_0=-\mathcal{E}$ and $p_3=L_z$ to the energy of the test particle and its projection of angular momentum on the z-axis.

Also, from the fulfillment of the circularity condition it follows that the $2$-planes orthogonal to the Killing vectors $X_i$ are integrable \cite{straumann}. Thus, there exist adapted coordinates $y^1$ and $y^2$ such that the metric tensor $g$ can be divided into two subspaces $g=\gamma\oplus G$, where $\gamma=\gamma_{ij}dx^i\otimes dx^j$ and $G=G_{AB}dy^A\otimes dy^B$. Furthermore, $g_{\mu\nu}$ depends only on the $y^A$ coordinates.

If the space-time $(M,g_{\mu\nu})$ admits a non-trivial quadratic first integral, then there exists a quadratic (or second-rank) Killing tensor $K^{\mu\nu}$ \cite{killing}. This tensor will yield a fourth constant of motion when contracted twice with the momenta $p_\mu$, that is, $K=K^{\mu\nu}p_\mu p_\nu$. The other three conserved quantities are the pair of momenta $p_i$, and the Hamiltonian of a freely falling test particle $2\mathcal{H}=g^{\mu\nu}p_\mu p_\nu=\kappa$, where $\kappa=0$ for null geodesics and $\kappa=-1$ for time-like geodesics.

Using the Hamilton-Jacobi equation it can be proven that the fourth conserved quantity comes from the separability of the Hamiltonian in two terms, each depending on the coordinates $x^A$ of some special coordinate system $\left\{t,x^1,x^2,\varphi\right\}$, and expressed as

\begin{equation}
 2\mathcal{H}=\kappa=[F_1(x^1)+F_2(x^2)]/[f(x^1)+h(x^2)],
\label{H}
\end{equation}

with $p_1=p_1(x^1)$ and $p_2=p_2(x^2)$. Since $g^{\mu\nu}p_\mu p_\nu=2\mathcal{H}$, equation (\ref{H}) constraints the form of the inverse metric tensor in the following way

\begin{equation}
g^{\mu\nu}=\frac{L^{\mu\nu}(x^1)+\Theta^{\mu\nu}(x^2)}{f(x^1)+h(x^2)},
\label{ginv}
\end{equation}

where $L^{\mu\nu}$ and $\Theta^{\mu\nu}$ are symmetrical tensors with the restriction $L^{2\mu}=L^{\mu2}=\Theta^{1\nu}=\Theta^{\nu1}=0$ and $f(x^1)$, $h(x^2)$ are one parameter functions. Notice that if the mentioned restriction on the symmetrical tensors would not be imposed, separability could not be achieved. Also note that the coordinates $x^A$ need not be the same as the previously introduced adapted coordinates $y^A$, however, it can be seen that the $\left\{x^\mu\right\}$ system can consist of adapted coordinates too. Suppose (\ref{H}) is not separable in the $y^A$ coordinates, then a change of basis from $y^A$ to $x^A$ using $y^A=y^A(x^1,x^2)$, would only affect the subspace of the metric orthogonal to both $X_i$ and hence, $x^A$ are still adapted coordinates. So, not any system of adapted coordinates will make equation (\ref{H}) separable, but those who do can also be adapted to the metric. Equation (\ref{H}) is separable too if a coordinate change of the form $x'^1=x'^1(x^1)$ and $x'^2=x'^2(x^2)$ is performed. Additionally, taking into account that $x^A$ are adapted coordinates of the metric, we have the further restriction on the symmetrical tensors that the only non-vanishing components of $L^{1\mu}=L^{\mu1}$ and $\Theta^{2\mu}=\Theta^{\mu2}$ are $L^{11}$ and $\Theta^{22}$, respectively. Adding up these restrictions, the symmetrical tensors can finally be written as $L=L^{ij}\partial_i\otimes\partial_j+L^{11}\partial_1\otimes\partial_1$ and $\Theta=\Theta^{ij}\partial_i\otimes\partial_j+\Theta^{22}\partial_2\otimes\partial_2$.

\end{proof}

It can also be shown that the Killing tensor is given by (see appendix A of \cite{DelAguila} for details)

\begin{equation}
K^{\mu\nu}=f(x^1)g^{\mu\nu}-L^{\mu\nu}(x^1)=\Theta^{\mu\nu}(x^2)-h(x^2)g^{\mu\nu},
\label{Killing}
\end{equation}

and so, the fourth constant of motion $K$ reads

\begin{equation}
K=f(x^1)\kappa-L^{ij}(x^1)p_ip_j-L^{11}(x^1)p_1^2=\Theta^{ij}(x^2)p_ip_j+\Theta^{22}(x^2)p_2^2-h(x^2)\kappa.
\label{K}
\end{equation}

Note that each equation of (\ref{K}) depends on a single coordinate.

\begin{lemma}
In a space-time $(M,g_{\mu\nu})$ with properties 1 to 4, the conformal factor $\tau$ can be defined to be positive definite in a time-like neighborhood $N$ of $\sigma$, and then $\left.\widetilde{g}^{AA}\right|_{N}>0$ for both $A=1,2$.
\end{lemma}

\begin{proof}
In a stationary, axially symmetric, and circular space-time with a time-like singularity $\sigma$, the vectors $\partial/\partial x^A$ are everywhere space-like in a neighborhood $N$ of $\sigma$. This implies that $\left.\nabla\tau\right|_N$ is space-like too. Using the adapted coordinates $\left\{x^\mu\right\}$ and the conformal form of the metric of property 4, we have $\left.g_{AA}\right|_N=\left.\tau/\widetilde{g}^{AA}\right|_N>0$. We can identify the quantities appearing in the inverse metric (\ref{ginv}) with the conformal factor and unphysical metric as $\tau=f(x^1)+h(x^2)$, and $\widetilde{g}^{\mu\nu}=L^{\mu\nu}(x^1)+\Theta^{\mu\nu}(x^2)$.

Now the lemma can be proven by contradiction. Consider a pair of points $x_1^A$ in $N$ such that $\tau(x_1^1,x_1^2)>0$. Then, $L^{11}(x_1^1),\Theta^{22}(x_1^2)\geq0$ because $N$ is a time-like neighborhood. Assume now there exist a different pair of points in $N$, say $x_1^1$ and $x_2^2$, for which $\tau(x_1^1,x_2^2)<0$. We now have that $g_{11}=\tau(x_1^1,x_2^2)/L^{11}(x_1^1)<0$ which clearly contradicts the hypothesis of $N$ being time-like. The same can be done for the $g_{22}$ component by considering other pair of points, $x_1^2$ and $x_2^1$, for which $\tau(x_2^1,x_1^2)<0$, thereby discarding also a change of sign of $\tau$ when keeping the point $x_1^2$ constant. As a result, we have that $\tau$ can be expressed as positive definite or negative definite, this implies that $\left.\widetilde{g}^{AA}\right|_{N}>0$ or $\left.\widetilde{g}^{AA}\right|_{N}<0$, respectively. We choose the positive definite option.

\end{proof}

Regarding the third property of the space-time, we can particularize the previously defined singular curvature set $\sigma_X$ to a stationary and axially symmetric space-time. Since $g_{\mu\nu}=g_{\mu\nu}(x^A)$, then  in the Cartesian-(like) coordinates $\{u^\mu\}$,

\begin{equation}
\sigma_X=\left\{\left(t,v\cos\varphi,v\sin\varphi,u^3\right)\in\mathbb{R}^4\mid 1/R_X(u^\mu)=0\right\},
\label{scs}
\end{equation}

with the quantities $v$ and $u^3$ depending only on the coordinates $x^A$. These coordinates need not be adapted to the metric. A point $q(x_0^A)\in\sigma_X$ can be expressed as $q(x_0^A)=(t,v_0\cos\varphi,v_0\sin\varphi,u_0^3)$, where $v_0=v(x_0^A)$ and $u_0^3=u^3(x_0^A)$. It is readily seen that if the  pair of points $x_0^A$ is unique for a given space-time we have that $$\sigma_X=\left\{(t,v_0\cos\varphi,v_0\sin\varphi,u_0^3)\mid-\infty<t<\infty,\,0\leq\varphi<2\pi\right\},$$ i.e., $\sigma_X=S^1\times\mathbb{R}$ and spatially compact provided that $v_0\neq0$. This will be the case for the class of metrics presented in section III.

The form of the inverse metric (\ref{ginv}) that resulted from lemma 1 can be utilized to compute the curvature scalars of the manifold. Their general expression is

\begin{equation}
R_X=F_X(L^{\mu\nu},\Theta^{\mu\nu},f,h)\Gamma^n/\tau^m,
\label{Ricci}
\end{equation}

where $F_X$ is a rather complicated function that depends on the curvature invariant, $\Gamma=det(\gamma)$ and $n,m\in\mathbb{Z}^+$. For the Ricci scalar, for example, $n=2$ and $m=7$. Examining (\ref{Ricci}) it can be observed that, if there exists a pair of points $x_0^A$ for which $\tau(x_0^A)=0$, then $q(x_0^A)\in\sigma_X$ and a curvature singularity can emerge in a common case. This hypothetical pair of points are later going to be of great relevance to the problem of affine completeness.

\begin{remarkk}
Other curvature singularities may arise apart from that of the pair $x_0^A$. For instance, the possible divergence of the determinant $\Gamma$ will yield another curvature singularity. This case will in general define a singular hyper-surface, and hence, $\sigma_X=\Sigma^2\times\mathbb{R}$, where $\Sigma^2$ is a two-manifold. See the metric of section IV for an explicit example of this. 
\end{remarkk}

\subsection{Geodesics Encountering the Curvature Singularity}

We are now ready to present the first theorem enlisting once again the properties of the space-times of interest. 

\begin{theorem}
Let $(M, g_{\mu\nu})$ be a four-dimensional space-time with $(-,+,+,+)$ signature and the following set of properties:
\begin{enumerate}
	\item Stationary, axially symmetric and satisfying the circularity condition
	\item Its geodesics admit a non-trivial quadratic first integral
	\item Contains a time-like curvature singularity $\sigma$ whose singular curvature set $\sigma_X$ is non-empty.
	\item There exists an unphysical space-time $(\tilde{M},\widetilde{g}_{\mu\nu})$ such that $M\cup\sigma_X\subseteq\tilde{M}$ and $g^{\mu\nu}=\widetilde{g}^{\mu\nu}/\tau$ with $\widetilde{g}^{\mu\nu},\tau\in C^{\infty}$ in a neighborhood $\tilde{N}\subset\tilde{M}$ of $\sigma_X$.
\end{enumerate}

Let also $X_i$ (with $i=0,3$) be the two Killing vectors related to property 1 of the metric, and $p_i$ their associated momenta. Define $\psi(p_i)=\tau(\kappa-g^{ij}p_ip_j)$, where $\kappa=0,-1$ for null and time-like geodesics, respectively. Then, at least a curve of the family $\eta(p_i)$ of causal geodesics in $N\subset M$, defined by a given pair $p_i\in\mathbb{R}$, will meet the singularity $\sigma$ if, and only if, starting from $n=0$, for any point $q\in\sigma_X$ and any $A=1,2$, there exists a first non-vanishing derivative $\left.\partial_A^{n}\psi(p_0,p_3)\right|_q$ such that either $n$ is odd, or the derivative is positive with $n$ even.
\end{theorem}

\begin{proof}
[Proof] From lemma 1 and the subsequent equation (\ref{K}) we can express the separated equations of motion in terms of the velocities $\dot{x}^1$ and $\dot{x}^2$,

\begin{eqnarray}
\left[(f+h)\dot{x}^1\right]^2&=&L^{11}(f\kappa-L^{ij}p_ip_j-K):=\Xi^1(x^1), 
\label{R} \\
\left[(f+h)\dot{x}^2\right]^2&=&\Theta^{22}(K+h\kappa-\Theta^{ij}p_ip_j):=\Xi^2(x^2),
\label{Xi}
\end{eqnarray}

here we have used $\dot{x}^A=g^{AA}p_A$. Note that trajectories defined by these equations of motion will only be possible for coordinate values such that $\Xi^A(x^A)\geq0$. At this point it will be helpful to introduce the following notation: $f(x^1)=j_1$, $h(x^2)=j_2$ and $\tau=j_1+j_2$. Please be aware that the superscript or subscript in the quantities $\Xi^A$ and $j_A$ is a tag for values $A=1,2$. However, it is not a tensorial index.

In the singularity $\sigma$, the equations of motion yield for some point $q(x_0^A)\in\sigma_X$:

\begin{eqnarray}
\Xi^1(x_0^1)&=&-L^{11}(x_0^1)(\alpha+K), \label{Ra}\\
\Xi^2(x_0^2)&=&-\Theta^{22}(x_0^2)(\beta-K),
\label{Xib}
\end{eqnarray}

with $\alpha=L^{ij}(x_0^1)p_ip_j-f(x_0^1)\kappa$ and $\beta=\Theta^{ij}(x_0^2)p_ip_j-h(x_0^2)\kappa$, which are quantities that are only in terms of parameters of the space-time (e.g. mass, angular momentum, etc.) and the constants of motion $p_i$. So, the curvature singularity will only be reached by geodesics if there exists a non-zero set of conserved quantities $p_0,p_3,K\in\mathbb{R}$ for which $\Xi^A(x_0^A)\geq0$ for both $A=1,2$.

By lemma 2 we have that\footnote{If $\tau$ were to be chosen as negative definite, the proof could carry on but with $L^{11}(x_0^1),\Theta^{22}(x_0^2)<0$. With this slight difference the theorem would still be valid, but with the final criterion for the first non-vanishing derivative $\left.\partial_A^{n}\psi(p_0,p_3)\right|_q$ changed to negative in the case of $n$ even.} $L^{11}(x_0^1),\Theta^{22}(x_0^2)>0$. Using these last conditions, one can easily realize that $\Xi^A(x_0^A)\geq0$ if, and only if, $\alpha+\beta<0$ for some real values of $p_0,p_3$. This last inequality can be rewritten as  

\begin{equation}
\left.\tau(g^{ij}p_ip_j-\kappa)\right|_q<0.
\label{crit0}
\end{equation}

It is important to remark that this is true because, for $\alpha+\beta<0$, one can always choose $K$ so that $\Xi^A(x_0^A)\geq0$ for both values of $A$. The singularity could appear, then, in the trajectory of a geodesic. On the contrary, if $\alpha+\beta>0$, there will not exist $K\in\mathbb{R}$ such that $\Xi^A(x_0^A)\geq0$ simultaneously and thus, an observer in geodesic motion does not reach the point $q$ of the singularity. Hereafter, we shall refer sometimes to these opposite situations as ``$\sigma$-encountering" and ``$\sigma$-avoidance", respectively. Note that the case $\alpha+\beta=0$ with arbitrary $K\neq-\alpha$ behaves similarly. This can be stated with all generality, except for the particular case $K=-\alpha=\beta$ which will be discussed later.  

We must now distinguish between two specific possibilities, namely $\tau(x_0^A)\neq0$ and $\tau(x_0^A)=0$. It turns out that in terms of geodesic incompleteness, the first one is practically trivial, while the second one proves to be far more interesting. Consider the equations of motion (\ref{R}) and (\ref{Xi}), as well as the two remaining ones $\dot{x}^i=\widetilde{g}^{ij}p_j/\tau$, it can be easily seen that

\begin{eqnarray}
\lim_{x^A\rightarrow x_0^A}(\dot{x}^A)^2&=&\begin{cases} 
      \mathcal{C}^A & \text{if } \tau(x_0^A)\neq0,\\
      \pm\infty & \text{if } \tau(x_0^A)=0,
      \end{cases}
\label{limxa} \\
\lim_{x^A\rightarrow x_0^A}\dot{x}^i&=&\begin{cases} 
      \mathcal{D}^i & \text{if } \tau(x_0^A)\neq0,\\
      \pm\infty & \text{if } \tau(x_0^A)=0,
      \end{cases}
\label{limxi}
\end{eqnarray}

where $\mathcal{C}^A=\Xi^A(x_0^A)/\tau(x_0^1,x_0^2)$ and $\mathcal{D}^i=\widetilde{g}^{ij}(x_0^A)p_j/\tau(x_0^1,x_0^2)$ are well-defined constants. Hence, even if both $\mathcal{C}^A>0$, the coordinate velocities of any geodesic approaching the singularity remain finite. The geodesic can be smoothly continued to future values of its affine parameter after, and despite, touching $\sigma$. In other words, the curvature singularity itself does not pose a threat to affine completeness (at least in a neighborhood $N$ of the curvature singularity) for the case $\tau(x_0^A)\neq0$. This will be called the trivially complete case. On the other hand, the second case could potentially lead to incompleteness. The signs of the limits taken in (\ref{limxa}) and (\ref{limxi}) depend clearly on the signs of $\Xi^A(x_0^A)$ and $\widetilde{g}^{ij}(x_0^A)p_j$ . If both $\Xi^A(x_0^A)>0$, not only geodesics do reach the singularity as mentioned before, but they also become incomplete.

\begin{remark}
This special situation needs to be considered since we have three independent conserved quantities that determine the motion of the test particle. Out of the four existing constants of motion, $\kappa$ is fixed depending on the nature of the geodesics (time-like or light-like), thus we are left with three degrees of freedom. This means we can impose restrictions on $p_i$ such that $-\alpha=\beta$ and then $K$ can be chosen to be equal to those expressions, leaving us with one undetermined conserved quantity. Despite this, and once the explicit restrictions are known, one should verify they correspond to physically realistic scenarios. The equations of motion for this case become

\begin{equation}
\left(\dot{x}^A\right)^2=\frac{\hat{\Xi}^A(x^A)}{\tau^2(x^1,x^2)}, 
\label{Xis}
\end{equation}

here we have defined $\hat{\Xi}^1(x^1):=L^{11}(f\kappa-L^{ij}p_ip_j+\alpha)$ and $\hat{\Xi}^2(x^2):=\Theta^{22}(h\kappa-\Theta^{ij}p_ip_j-\alpha)$. It is clear that, in the trivially complete case, there will be no trouble at all evaluating the limits (\ref{limxa}) and one simply obtains $\dot{x}^A=0$. Depending on the second derivative $\ddot{x}^A$, either $q(x_0^A)$ will be a turning point or motion will be constrained to it. No further discussion is needed and inequality (\ref{crit0}) alone suffices to determine if geodesics avoid or not the point $x_0^A$ of the singularity.

For the case $\tau(x_0^A)=0$, the analysis increases in complexity but similar conclusions can be drawn. As the singularity $\sigma$ is being approached, i.e. as $x^A\rightarrow x_0^A$, the functions $\hat{\Xi}^A(x_0^A)\rightarrow 0$ and $\tau(x_0^1,x_0^2)\rightarrow 0$. So, the limit for equations (\ref{Xis}) is undetermined. We must be careful on this matter since it could be possible that said limit does not exist, that is, it depends on the path taken to reach the singularity. The simplest way to verify this is to expand the functions $\hat{\Xi}^A$ in power series around the discussed point of the singularity, thus describing their behavior in a neighborhood of it: $\hat{\Xi}^A=\sum_nc_n^A(x^A-x_0^A)^n$. Observe that if the leading term of a given series $\hat{\Xi}^A$ is odd, then the limit we are computing does not exist since $\lim_{x^A\rightarrow x_0^A}(\dot{x}^A)^2=\pm\infty$, with the sign depending on the way $x_0^A$ is being approached. An explicit expression for the coefficients $c_n^A$ shall be given forward. It is worth mentioning too that one can always find a well-defined series expansion of these functions, as well as any other depending on $\widetilde{g}^{\mu\nu}$ and $\tau$, due to property 4 of the space-time.

The two cases of motion constrained to each $x^A=x_0^A$ plane ($\dot{x}^A=\ddot{x}^A=0$) deserve particular attention. The necessary conditions for this to happen are $K=-\alpha$, $c_1^1=0$ for $A=1$, and $K=\beta$, $c_1^2=0$, for $A=2$. Note that both situations do not represent any problem for $\alpha\neq\beta$ and inequality (\ref{crit0}) suffices to determine if there exists any geodesic that touches the singularity. Nevertheless, the case $K=-\alpha=\beta$ results in an undetermined limit for (\ref{Xis}) as previously explained. Also, the remaining undetermined constant of motion can be used to satisfy either $c_1^1=0$ or $c_1^2=0$. 

Restricting motion to $x^2=x_0^2$, the leading term in $\hat{\Xi}^1$ will be $c_1^1=\left[-L^{11}\left.\partial_1(L^{ij}p_ip_j)\right]\right|_q$. To guarantee that geodesics are unable to reach the point $x_0^A$ of the singularity this quantity must vanish identically, we then are left with the second order coefficient as the leading term $c_2^1=\left[L^{11}\partial_1^2\left.(j_1\kappa-L^{ij}p_ip_j)\right]\right|_q$. The requirement that $c_1^1$ vanishes implies local symmetry of $\hat{\Xi}^1$ about $x_0^1$. If $c_1^2<0$, there will be $\sigma$-avoidance, if $c_2^1>0$ there will be $\sigma$-encountering. However, consider a pair of constants of motion $p_i$ chosen so that $c_2^1=0$. This is still possible because, as $c_1^1$ vanishes identically, we still have one degree of freedom for those conserved quantities. For this case, the $n=3$ coefficient would serve as a criterion to determine whether geodesics do reach or not $\sigma$. The same characteristics described for the first-order term of the expansion apply for $n=3$, in fact, for any $n$ odd. In a similar manner, the characteristics described for the second-order term apply for any $n$ even. In general, the $n$-th coefficient of this series will be given by

\begin{equation}
c_n^1=\frac{1}{n!}L^{11}\partial_1^n\left.(j_1\kappa-L^{ij}p_ip_j)\right|_q,
\label{crit1n}
\end{equation}

as long as the preceding terms vanish, that is, $c_m^1=0$ for all $m<n$ with $m,n\in\mathbb{Z}^+$. The same analysis can be done when constraining motion to the plane $x^1=x_0^1$. The coefficients $c_n^2$ for the function $\hat{\Xi}_2$ are

\begin{equation}
c_n^2=\frac{1}{n!}\Theta^{22}\partial_2^n\left.(j_2\kappa-\Theta^{ij}p_ip_j)\right|_q,
\label{crit2n}
\end{equation}

as long as the preceding terms vanish. With this, we end the discussion of the considerations that need to be taken into account for the case $K=-\alpha=\beta$.
\end{remark}

Together, the expressions given by (\ref{crit0}), (\ref{crit1n}) and (\ref{crit2n}) determine if geodesics are able to reach the curvature singularity $\sigma$, they may be rewritten in the general form

\begin{equation}
\left.\partial_A^n\left[\tau(\kappa-g^{ij}p_ip_j)\right]\right|_q>0.
\label{critg}
\end{equation}

We may regard the set of all causal geodesics in the space-time $(M, g_{\mu\nu})$ as a six-parameter family of curves\footnote{Since there are two Killing vectors $X_i$, the number of parameters is reduced from the standard $8$ (the set of initial conditions $x^\mu(0)$ and $\dot{x}^\mu(0)$) to $6$, due to the isometries of the space-time in the directions $x^i$.}. For convenience, these parameters are chosen to be the conserved quantities $p_i,K\in\mathbb{R}$, $\kappa=0,-1$, and a pair of initial conditions $x^A(0)\in\mathbb{R}$. Fixing the values of momenta $p_i$ hence, defines a 4-parameter subfamily $\eta(p_i)$ of causal geodesics. So, making $\psi(p_0,p_3)=\tau(\kappa-g^{ij}p_ip_j)$ and summing up the above analysis, we find that at least a curve of the subfamily $\eta(p_i)$ will encounter the curvature singularity if, and only if, starting from $n=0$, for any point $q\in\sigma_X$ and any $A=1,2$, there exists a first non-vanishing derivative $\left.\partial_A^{n}\psi(p_0,p_3)\right|_q$ such that either $n$ is odd, or the derivative is positive with $n$ even.

This concludes the proof.
\end{proof}

\subsection{Complete Geodesics Going Through the Singularity}

By now we have established the conditions for which causal geodesics in certain stationary and axially symmetric space-times meet the curvature singularity within them. This unfortunate fate is a common source of affine incompleteness and thus, of possible ill-behavior of causal curves in the space-time. Here, we shall discuss some cases in which reaching the curvature singularity does not necessarily imply a breakdown of completeness. This possibility was previously mentioned for the trivially complete case and now we extend it to the $\tau(x_0^A)=0$ case.

\begin{theorem}

Let $(M, g_{\mu\nu})$ be a space-time with properties 1 to 4 and $\eta'(p'_i)$ a family of causal geodesics that encounter the curvature singularity at a single point $x_0^A$ ($q(x_0^A)\in\sigma_X$). Consider an expansion in power series of $\tau(x^A)$ around the singular point $x_0^A$, and let $\delta_A$ be the order of the leading terms that go as $(x^A-x_0^A)^{\delta_A}$ in the series for each $A$. In a sufficiently small time-like neighborhood $N$ of $\sigma$, the curves of $\eta'(p'_i)$ can be smoothly continued to subsequent values of their corresponding affine parameter after encountering $\sigma$, if for some momenta $p'_i\in\mathbb{R}$ the following holds:

\renewcommand{\labelenumi}{\roman{enumi}}
\begin{enumerate}
\item $\left.\partial_A^{n}\psi(p'_0,p'_3)\right|_q=0$ for $0\leq n<2\delta_A$ and both $A=1,2$, 
\item $\left.\partial_A^n\widetilde{g}^{ij}\right|_qp'_j=0$ for $0\leq n<\delta_A$, both $A=1,2$ and both $i=0,3$.
\end{enumerate}

\end{theorem}

\begin{proof}
We now focus on a power series expansion of the function $\tau=\sum_{A,n}d^A_n(x^A-x_0^A)^n$, where we have written explicitly the sum on the $x^A$ coordinates to avoid any sort of confusion. Since $\tau$ is positive definite in a time-like neighborhood $N$, it follows that the leading terms $d_n^A$ of the series are positive with $n$ even. Let $\delta_A$ be the order of the leading term of the $j_A$ functions and define for compactness $\mathcal{B}_A=d_{\delta_A}^A$. Assume for the affine parameter that $x^A\rightarrow x_0^A$ as $\lambda\rightarrow\lambda_0$. Then, consider the following limit

\begin{equation}
\lim_{\lambda\rightarrow\lambda_0}\frac{(x^A-x_0^A)^n}{\left[\mathcal{B}_1(x^1-x_0^1)^{\delta_1}+\mathcal{B}_2(x^2-x_0^2)^{\delta_2}\right]^{n'}}=\begin{cases}
      \text{undetermined} & \text{for } n<n'\delta_A,\\ 
      w/\mathcal{B}_A^{n'} & \text{for } n=n'\delta_A,\\
      0 & \text{for } n>n'\delta_A,
      \end{cases}
\label{limxn}
\end{equation}

with $w\in[0,1]$ and $n,n'\in\mathbb{Z}^+$. In computing these limits we have used the fact that $\mathcal{B}_A>0$ and $\delta_A$ is even.

Examining equation (\ref{limxn}) along with (\ref{Xis}), and since the limit shown there is bounded for $n\geq2\delta_A$, a hint towards possible complete geodesics going through $\sigma$ can be found. Consider a family of causal geodesics $\eta'(p_i)$ with fixed momenta such that for $0\leq n<2\delta_A$ and both $A=1,2$, the function $\left.\partial_A^{n}\psi(p_0,p_3)\right|_q$ vanishes in some singular point $q(x_0^A)\in\sigma_X$. Hence, as a geodesic of $\eta'(p_i)$ approaches $x_0^A$, its coordinate velocities $\dot{x}^A$ will be bounded. We must not forget, though, about the other pair of equations of motion $\dot{x}^i=\widetilde{g}^{ij}p_j/\tau$ which could still cause incompleteness.

For this purpose we perform yet another expansion in power series around the points $x_0^A$. In this last case for $\widetilde{g}^{ij}p_j=\sum_{A,n}b^{Ai}_n(x^A-x_0^A)^n$, where

\begin{equation}
b^{Ai}_n=\frac{1}{n!}\left.\partial_A^n\widetilde{g}^{ij}\right|_qp_j.
\label{bain}
\end{equation}

By similar arguments as those used for the $\dot{x}^A$ equations and using once again the limit shown in (\ref{limxn}) with $n'=1$, it can be seen that if the coefficients $b_n^{Ai}$ vanish for $A=1,2$, $i=0,3$ and $0\leq n<\delta_A$, the coordinate velocities $\dot{x}^i$ will be bounded while approaching the singularity. So, in a sufficiently small time-like neighborhood $N$ of it, the causal geodesics of a family $\eta'(p'_i)$ will encounter the curvature singularity while still being able to be smoothly continued to subsequent values of their affine parameter, if for some momenta $p'_i\in\mathbb{R}$ the following holds:

\renewcommand{\labelenumi}{\roman{enumi}}
\begin{enumerate}
\item $\left.\partial_A^{n}\psi(p'_0,p'_3)\right|_q=0$ for $0\leq n<2\delta_A$ and both $A=1,2$, ,
\item $\left.\partial_A^n\widetilde{g}^{ij}\right|_qp_j=0$ for $0\leq n<\delta_A$, both $A=1,2$ and both $i=0,3$.
\end{enumerate}

This concludes the proof.

\end{proof}

Note that in this theorem the trivially complete case, i.e. $\tau\neq0$ and hence $\delta_A=0$, is included.

Also, the second condition of the theorem can be expressed in an alternative way. The vanishing of $\left.\partial_A^{n}\psi(p_0,p_3)\right|_q$ for $0\leq n<2\delta_A$ and both $A=1,2$, yields the following relation for the momenta $p_0$ and $p_3$,

\begin{equation}
\left[\left.\partial_A^n(\widetilde{g}^{ii}p_i+\widetilde{g}^{ij}p_j)\right|_q\right]^2=\left.\left[-\det(\partial_A^n\widetilde{\gamma}^{-1})p_j^2+\kappa(\partial_A^n\widetilde{g}^{ii})(\partial_A^n\tau)\right]\right|_q,
\label{pipj}
\end{equation}

with $\det(\partial_A^n\widetilde{\gamma}^{-1})=(\partial_A^n\widetilde{g}^{00})(\partial_A^n\widetilde{g}^{33})-(\partial_A^n\widetilde{g}^{03})^2$. Warning: in equation (\ref{pipj}) we have temporarily abandoned the summation convention for repeated indices, the intended use of this expression is for fixed values $i,j=0,3$ and $i\neq j$. This liberty is taken only in this equation and in the following one. This particular form of the equation is used due to it being easily substituted in (\ref{bain}), obtaining thus

\begin{equation}
b^{Ai}_n=\left.\pm\frac{1}{n!}\sqrt{-\det(\partial_A^n\widetilde{\gamma}^{-1})p_j^2+\kappa(\partial_A^n\widetilde{g}^{ii})(\partial_A^n\tau)}\right|_q.
\label{bain2}
\end{equation}

The advantage of using (\ref{bain2}) over (\ref{bain}) lies in the reduction of one free parameter in the equation (one of the constants of motion), despite this, the latter is way more compact.

The result of theorem 2 states the sufficient conditions for which causal geodesics may be continued to future values of its affine parameter after meeting a point of the curvature singularity. Nevertheless, if the space-time does not contain those future points, said curve would still disappear off the manifold in a finite amount of said parameter, and consequently, be incomplete. To avoid incompleteness, at least in the neighborhood of the singularity, an additional requirement may be that of the existence of a space-time extension containing the future points of the geodesic. 

\begin{corollary}
Let $(M, g_{\mu\nu})$ be a space-time with properties 1 to 4, and let $\xi(\lambda)\in\eta'(p'_i)$ be a curve with affine parameter $\lambda$ that encounters a single singular point $q(x_0^A)\in\sigma_X$ at a parameter value $\lambda=\lambda_0$. The curve $\xi$ will not become incomplete in a sufficiently small time-like neighborhood $N$ of $\sigma$ if conditions i and ii are satisfied, and provided there exists a suitable extension of $M$ that contains the future points of $\xi(\lambda_0+\epsilon)$ for small $\epsilon$.
\end{corollary}

This means that, given the mentioned conditions, $\xi$ enters the time-like neighborhood $N$ and can then be extended to large enough values of its affine parameter $\lambda$ until the curve either leaves $N$, or stays in it for all subsequent values $\lambda>\lambda_0$. It could also happen that the curve remains within the neighborhood for all $\lambda\in\mathbb{R}$.

\section{III. The Pleba\'nski-Demia\'nski Class of Space-times}

In this section we apply our results to the Pleba\'nski-Demia\'nski class of space-times \cite{plebanski}. This class consists of solutions to the Einstein-Maxwell field equations with a generally non-zero cosmological constant. The physically relevant space-times of this class describe black holes with up to seven parameters, namely:
\begin{itemize}
\item The mass $m$ of the black hole
\item Angular momentum per unit mass $a$
\item Electric charge $Q$
\item The Taub-NUT parameter $l$
\item The Manko-Ruiz constant $C$ which is only relevant when $l\neq0$
\item The acceleration of the black hole $\mathcal{A}$
\item Cosmological constant $\Lambda$
\end{itemize}

We will limit ourselves to asymptotically flat space-times, these are only possible if $\Lambda=\mathcal{A}=0$. The reason for this is that, in the general case, there can only exist at most a conformal Killing tensor for metrics with $\Lambda,\mathcal{A}\neq0$, allowing the integrability of the equations of motion for null geodesics only. This falls beyond the scope of the derived theorems. In Boyern-Lindquist coordinates, the line element of the asymptotically flat metrics is then given by \cite{griffiths}

\begin{equation}
ds^2=-\frac{1}{\Sigma}\left[(\Delta_r-a^2\sin^2\theta)dt^2+(\Delta_r\chi-a(\Sigma+a\chi)\sin^2\theta)dtd\varphi+\left((\Sigma+a\chi)^2\sin^2\theta-\Delta_r\chi^2\right)d\varphi^2\right]+\frac{\Sigma}{\Delta_r}dr^2+\Sigma d\theta^2,
\label{pd}
\end{equation}

where $\Sigma=r^2+(l+a\cos\theta)^2$, $\chi=a\sin^2\theta-2l(\cos\theta+C)$ and $\Delta_r=r^2-2mr+a^2-l^2+Q^2$. By setting the appropriate parameters to zero in (\ref{pd}) we can obtain some thoroughly studied metrics such as Kerr, Reissner-Nordström, Schwarzschild, etc.  The event horizons of the black hole are located at the roots of $\Delta_r$, which yield $r_{1,2}=m\pm\sqrt{m^2-(a^2-l^2+Q^2)}$. These space-times contain a ring singularity when $\Sigma=0$, i.e. $r=0$ and $\cos\theta=-l/a$. The singular curvature set for this class of space-times can be expressed using the Cartesian-like Kerr-Schild coordinates\footnote{This set of coordinates is particularly helpful in realizing that the singularity of this type of black holes is indeed a ring \cite{KerrSchild}.} for which $v=\sqrt{r^2+a^2}\sin\theta$ and $u^3=r\cos\theta$ in equation (\ref{scs}). Since the singularity is defined by a single pair of points, $r=0$ and $\cos\theta=-l/a$, then $\sigma_X$ can be written as $$\sigma_X=\left\{\left(t,v_0\cos\varphi,v_0\sin\varphi,0\right)\mid-\infty<t<\infty, 0\leq\varphi<2\pi\right\},$$ where $v_0=\sqrt{a^2-l^2}$. If $a^2-l^2+Q^2>m^2$, the event horizons do not exist and the ring singularity is, in principle, left visible to any asymptotically distant observer. Also, note that the curvature singularity will only exist if $\left|a\right|>\left|l\right|$. This inequality is consistent with the no event horizon condition. The ring singularity will always be time-like, except for the case $a=l=Q=0$, which reduces (\ref{pd}) to the Schwarzschild metric. The Schwarzschild singularity $r=0$ is known to be space-like and hence, once a particle crosses the event horizon, the singularity will unavoidably appear in the future of its world-line. 

It is simple to compute the inverse metric $g^{\mu\nu}$ of this class of space-times and express it in the separated form

\begin{equation}
g^{\mu\nu}=\frac{\mathcal{R}^{\mu\nu}(r)+\Theta^{\mu\nu}(\theta)}{f(r)+h(\theta)},
\label{ginvPD}
\end{equation}

explicitly we have, 

\begin{eqnarray}
\mathcal{R}^{\mu \nu}=
\begin{bmatrix}
	-(\Sigma+a\chi)^2/\Delta_r & 0 & 0 & -a(\Sigma+a\chi)/\Delta_r \\
	0 & \Delta_r & 0 & 0 \\
	0 & 0 & 0 & 0 \\
	-a(\Sigma+a\chi)/\Delta_r & 0 & 0 & -a^2/\Delta_r \\ 	
\end{bmatrix}, \hspace{0.5cm}
\Theta^{\mu \nu}=
\begin{bmatrix}
	\chi^2/\sin^2\theta & 0 & 0 & \chi/\sin^2\theta \\
	0 & 0 & 0 & 0 \\
	0 & 0 & 1 & 0 \\
	\chi/\sin^2\theta & 0 & 0 & -1/\sin^2\theta \\ 	
\end{bmatrix},
\label{RThuv}
\end{eqnarray}

with $f(r)=r^2$ and $h(\theta)=(l+a\cos\theta)^2$. Observe that the sum $\Sigma+a\chi$ depends only on the coordinate $r$.

One of the most important features of this class of solutions is that they are algebraically special in the Petrov classification of space-times. In particular, they are of type D, meaning they possess two principal null directions which are repeated twice \cite{stephani}. This implies the existence of a second rank Killing tensor and thus, the integrability of the equations of motion in the space-time \cite{killing}. The Killing tensor has the structure of equation (\ref{Killing}).

Since the Pleba\'nski-Demia\'nski class of space-times satisfies the properties stated in theorem 1, we can check if its causal geodesics reach the ring singularity. Furthermore, if they do reach it, we can establish by using theorem 2 if they remain complete as they pass through it. Evaluating $\psi=\Sigma(\kappa-g^{ij}p_ip_j)$ in the ring singularity, we obtain

\begin{equation}
\left.\psi\right|_q=-\frac{Q^2\left((a^2-2alC+l^2)\mathcal{E}-aL_z\right)^2}{(a^2-l^2)(a^2-l^2+Q^2)},
\label{disPD}
\end{equation}

where we have used $p_0=-\mathcal{E}$ and $p_3=L_z$. At this point, it shall be useful to consider two different cases of these class of space-times as they will possess some unique interesting properties.

\subsubsection{Case $Q\neq0$}

For a non-vanishing electric charge, clearly $\left.\psi\right|_q$ is strictly negative and so, there will be $\sigma$-avoidance for geodesics of arbitrary values of energy and angular momentum. Nonetheless, consider the particular case $(a^2-2alC+l^2)\mathcal{E}=aL_z$ which makes $\left.\psi\right|_q$ zero. We therefore have to analyze the first non-vanishing derivatives of $\psi$. It can be seen that $\left.\left(\partial_r\psi\right)\right|_q=\left.\left(\partial_\theta\psi\right)\right|_q=0$ when the previously mentioned condition on energy and angular momentum holds. For the second derivatives we have

\begin{equation}
\left.\left(\partial_r^2\psi\right)\right|_q=\frac{1}{a^2-l^2}\left.\left(\partial_\theta^2\psi\right)\right|_q=2\kappa.
\label{dis2DP}
\end{equation}

From (\ref{dis2DP}) it can be concluded that time-like geodesics ($\kappa=-1$) do not reach the singularity in asymptotically flat space-times of the Pleba\'nski-Demia\'nski class. This is in full agreement with one of the results from \cite{geroch}, where the singular region of a Reissner-Nordström black hole is shown to be physically inaccessible to time-like curves of limited acceleration. However, for null geodesics the second derivatives reduce to zero once again. The same goes for the third derivatives $\left.\left(\partial_r^3\psi\right)\right|_q=\left.\left(\partial_\theta^3\psi\right)\right|_q=0$. It is until the fourth derivation that a non-vanishing constant can be found,

\begin{equation}
(a^2-l^2+Q^2)\left.\left(\partial_r^4\psi\right)\right|_q=-\frac{1}{a^2-l^2}\left.\left(\partial_\theta^4\psi\right)\right|_q=4!\mathcal{E}^2.
\label{dis4DP}
\end{equation}

Hence, there will be $\sigma$-encountering for light-like geodesics of the asymptotically flat ($\mathcal{A}=\Lambda=0$) Pleba\'nski-Demia\'nski class of space-times. The equations of motion for this case are simply

\begin{equation}
(\Sigma\dot{r})^2=\mathcal{E}^2r^4, \hspace{0.5cm} (\Sigma\dot{\theta})^2=-\mathcal{E}^2(l+a\cos\theta)^4/a^2.
\label{rth}
\end{equation}

It can be easily seen that motion is only possible within the plane $\cos\theta=-l/a$. A null observer constrained to this plane, and approaching from infinity, will eventually touch the ring singularity. For this to be the case, specific conditions on the constants of motion have to be met, namely $(a^2-2alC+l^2)\mathcal{E}=aL_z$, $\kappa=K=0$, which define a family $\eta(\mathcal{E},L_z,K=0)$ of null geodesics. Since the curves of this family are constrained to the discussed plane, $\eta$ is a one-parametric family of null geodesics, being the initial condition $r(0)$ the only degree of freedom. These curves are the principal null rays, the rest of the causal geodesics will never reach the singularity. It is worth pointing out that the $\sigma$-encountering of the principal null rays of the Kerr metric was already mentioned in \cite{boyer}. 

\subsubsection{Case $Q=0$}

Focusing now on the metrics (\ref{pd}) without electric charge, it can be seen in (\ref{disPD}) that $\left.\psi\right|_q$ is zero without imposing restrictions on the conserved quantities $\mathcal{E}$ and $L_z$. In this case the first non-vanishing derivatives are

\begin{equation}
\frac{a^2-l^2}{2m}\left.\left(\partial_r\psi\right)\right|_q=\frac{\sqrt{a^2-l^2}}{2l}\left.\left(\partial_\theta\psi\right)\right|_q=\frac{\left((a^2-2alC+l^2)\mathcal{E}-aL_z\right)^2}{a^2-l^2}.
\label{dis1Q0}
\end{equation}

These are derivatives of odd order and thus, there are causal geodesics that unfortunately do reach the ring singularity. Assuming a positive mass, as is the case for physical black holes, these geodesics are the ones coming (going) from (to) $r>0$ and $\theta>\arccos(-l/a)$. Hence, they remain within the inner region of the black hole ($0\leq r<m+\sqrt{m^2-a^2+l^2}$) with no chance of escaping to the domain of outer communications. For the region of negative values of $r$, which is obtained by performing a maximal analytic extension of the space-time manifold for metrics with $a\neq0$ or $l\neq0$ (similar to that of \cite{boyer}), we have the opposite situation and causal geodesics will actually be repelled from the singularity.

One may still try to obtain causal geodesics that meet $\sigma$ coming from either side of it by setting the conjugate momenta to $(a^2-2alC+l^2)\mathcal{E}=aL_z$. This would result in the vanishing of the first derivatives in (\ref{dis1Q0}), leading then to the same results shown in equations (\ref{dis2DP}) and (\ref{dis4DP}), but with $Q=0$.

\subsection{Completeness of Null Geodesics Encountering the Singularity}

We have already shown that there indeed exists a family of null geodesics which meet the singularity in the Plebanski-Demianski class of space-times. In what follows we wonder if those curves can yet remain complete despite this fact. For this purpose we apply theorem 2.

Consider the family $\eta(\mathcal{E},L_z,K=0)$ of null geodesics with $(a^2-2alC+l^2)\mathcal{E}=aL_z$. In the above calculations we saw that for this family, the first non-vanishing derivatives of $\left.\partial_A^{n}\psi\right|_q$ were those of order $n=4$ for both coordinates $r$ and $\theta$. Since $\delta_r=\delta_\theta=2$ because $\Sigma=r^2+(l+a\cos\theta)^2$, condition i of theorem 2 is satisfied. To verify condition ii, all that is left to be done is to compute the $\left.\partial_A^n\widetilde{g}^{ij}\right|_qp_j$ derivatives. An evaluation of these expressions yields

\begin{eqnarray}
\frac{1}{a^2-2alC+l^2}\left.\widetilde{g}^{0j}\right|_qp_j&=&\frac{1}{a}\left.\widetilde{g}^{3j}\right|_qp_j=-\frac{Q^2\left((a^2-2alC+l^2)\mathcal{E}-aL_z\right)}{(a^2-l^2)(a^2+Q^2-l^2)} \nonumber\\
\frac{1}{a^2-2alC+l^2}\left.\partial_r\widetilde{g}^{0j}\right|_qp_j&=&\frac{1}{a}\left.\partial_r\widetilde{g}^{3j}\right|_qp_j=\frac{2m\left((a^2-2alC+l^2)\mathcal{E}-aL_z\right)}{(a^2+Q^2-l^2)^2} \nonumber\\
\frac{1}{a^2-2alC+l^2}\left.\partial_\theta\widetilde{g}^{0j}\right|_qp_j&=&\frac{1}{a}\left.\partial_\theta\widetilde{g}^{3j}\right|_qp_j=-\frac{2l\left((a^2-2alC+l^2)\mathcal{E}-aL_z\right)}{(a^2-l^2)^{3/2}}.
\label{gijpj}
\end{eqnarray}

Condition ii is clearly satisfied too for the family $\eta(\mathcal{E},L_z,K=0)$ of null geodesics. Hence, these curves can be smoothly continued to future values of its affine parameter after touching the ring singularity. Indeed, from the equations of motion (\ref{rth}) and the fact that these causal geodesics are constrained to the plane $\cos\theta=-l/a$, we have for the radial coordinate velocity $\dot{r}=\pm\mathcal{E}$, which describes ingoing and outgoing radial null rays. This equation is easily integrated and shows that a geodesic coming from positive (negative) $r$ touches the curvature singularity $\sigma$, and then continues its trajectory to negative (positive) $r$. Hence, despite there being $\sigma$-encountering, the ring singularity does not lead to incomplete curves for the case of the null geodesics of the family $\eta$. It is remarkable that the tangents of these geodesics are aligned with the principal null directions $k^\mu_\pm$ of the space-time, namely

\begin{equation}
k^\mu_\pm=\left[(r^2+a^2-2alC+l^2)/\Delta_r,\pm1,0,a/\Delta_r\right],
\label{pnd}
\end{equation}

expressed in the $\{t,r,\theta,\varphi\}$ basis of Boyer-Lindquist coordinates. 

It is worth now summarizing the results here obtained for each case. In an asymptotically flat space-time with non-vanishing electric charge of the Pleba\'nski-Demia\'nski class, there will be in general $\sigma$-avoidance for time-like geodesics, while the principal null rays shall indeed meet the ring singularity. On the other hand, for metrics with $Q=0$, we found that there can be $\sigma$-encountering for causal geodesics with arbitrary values of momenta coming from $r>0$. These curves will reach the ring singularity in a finite amount of its affine parameter and then disappear off the manifold, the reverse situation is possible as well, that is, an observer could suddenly appear in the singularity and then follow its way into the manifold. Nevertheless, in both cases of electric charge value and $a\neq0$, the principal null rays meet the singularity and continue their paths into the region of opposite sign of radial coordinate values. In other words, the singularity does not induce incompleteness in those curves as they can be extended to future values of their affine parameter after encountering it.

Finally, causal geodesic completeness may be confirmed for a particular metric of the class of space-times studied in this section, namely the Kerr-Newman metric. This is due to the fact that its maximal analytic extension is known to include negative values of $r$, and also that it does not possess any other kind of singularity within its space-time manifold. Both of these properties are also shared with the Kerr black hole, but since $Q=0$ for this case, there are incomplete causal geodesics as mentioned before. Unfortunately, in space-times with $l\neq0$, e.g. Kerr-NUT, a conical singularity in the symmetry axis is formed which can provoke geodesic incompleteness, this falls beyond the scope of our paper. In this cases, our analysis can at most tell that geodesics do not become incomplete inside a sufficiently small neighborhood of the ring singularity. The Reissner-Nordström metric can be discarded as well from being geodesically complete as it cannot be extended to negative $r$.

\section{IV. A Geodesically Complete Wormhole Space-time}

In the previous section we analyzed a physically relevant class of space-times and found that the singularity lies in the path of some null geodesics with specific constraints on the constants of motion. It might be natural now to wonder if a space-time whose singularity is inaccessible for both, null and time-like geodesics, does actually exist. In this spirit, we propose an example of such a case in what follows.

The line element is roughly based on the axially symmetrical wormhole (WH) space-times (specifically the so-called ring WHs) found in \cite{DelAguila,Miranda}, but with the element $g^{tt}$ of the inverse metric tensor modified so that the property $\left.\psi\right|_q<0$ holds for any conserved quantity $p_i$. Evidently, the construction of this metric is guided by geometrical arguments rather than physical significance. Hence, the gravitational source that could produce such a space-time geometry could not bare any physical relevance. We use spheroidal oblate coordinates $x$, $y$ to express the corresponding line element,

\begin{equation}
ds^2=-\frac{\Delta^2}{\Delta_s}dt^2+L^2\frac{\Delta}{\Delta_1}dx^2+\frac{\Delta}{1-y^2}dy^2-2aLx(1-y^2)\frac{\Delta}{\Delta_s}dtd\varphi+(1-y^2)(\Delta\Delta_1-a^2)\frac{\Delta}{\Delta_s}d\varphi^2,
\label{ds2WH3}
\end{equation}

where we have defined $\Delta=L^2(x^2+y^2)$, $\Delta_1=L^2(x^2+1)$ and $\Delta_s=\Delta^2-a^2y^2$. This set of coordinates is related to those of Boyer-Lindquist through $Lx=r-r_1$ and $y=\cos\theta$. Also, $L$ is defined as $L^2=r_0^2-r_1^2$ with $r_0$ and $r_1$ being constant length parameters, and $a$ a parameter with units of angular momentum. The inverse metric $g^{\mu\nu}$ has the more compact expression $g^{\mu\nu}=[\mathcal{X}^{\mu\nu}(x)+\mathcal{Y}^{\mu\nu}(y)]/[f(x)+h(y)]$. These functions are given by $f(x)=L^2x^2$ and $h(y)=L^2y^2$, while the tensors $\mathcal{X}^{\mu \nu}$ and $\mathcal{Y}^{\mu \nu}$ by

\begin{eqnarray}
\mathcal{X}^{\mu \nu}&=&
\begin{bmatrix}
	-L^2x^2+a^2/\Delta_1 & 0 & 0 & -aLx/\Delta_1 \\
	0 & \Delta_1/L^2 & 0 & 0 \\
	0 & 0 & 0 & 0 \\
	-aLx/\Delta_1 & 0 & 0 & -L^2/\Delta_1 \\ 	
\end{bmatrix}, \nonumber \\
\mathcal{Y}^{\mu \nu}&=&diag\left[-L^2y^2,0,1-y^2,\frac{1}{1-y^2}\right].
\label{XYuv}
\end{eqnarray}

We now present the embedding profiles of the metric in three-dimensional Euclidean space. These show that the line element (\ref{ds2WH3}) has indeed a wormhole geometry whose throat is a disc of radius $L$ located at $x=0$. The throat connects two different universes (or possibly distant regions of the same universe), one with $x>0$ and another with $x<0$.

\begin{figure}[htp]
	\centering
		\includegraphics[scale=0.8]{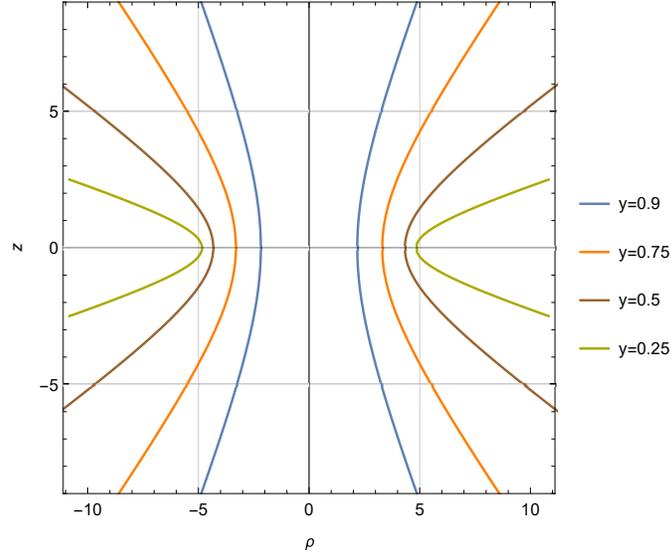} 
		\caption{Embedding diagram of the WH in three-dimensional Euclidean space for different constant values of $y$ with $a=0.1$ and $L=5$. Here, $z$ and $\rho$ are the usual cylindrical coordinates.}
	\label{fig:geo}
\end{figure}

The singular regions of this metric can be revealed by computing the Ricci scalar, which is

\begin{equation}
R=\frac{-3a^2L^2(3a^2y^2+\Delta^2)(y^2(1-y^2)+x^2(1+3y^2))}{2\Delta^2\Delta_s^2}.
\label{RicciWH}
\end{equation} 

From (\ref{RicciWH}) we notice that $y=1$ is nothing but a coordinate singularity due to the choice of our spheroidal coordinates. The root $\Delta=0$ corresponds to a ring singularity $\sigma$ ($x=y=0$) similar to that of the Pleba\'nski-Demia\'nski class of space-times. While $\Delta_s=0$ yields an additional singularity $\sigma'$ with very interesting properties that shall be discussed later in this section. For the time being, we focus our attention on the ring singularity. 

We proceed by calculating the quantity $\psi=\Delta(\kappa-g^{ij}p_ip_j)$ in $x=y=0$, thus finding

\begin{equation}
\left.\psi\right|_q=-(a\mathcal{E}/L)^2.
\label{disWH}
\end{equation}  

For non-zero values of energy, any observer traveling in geodesic motion will be repelled from the ring singularity. Following the same procedure as in the Pleba\'nski-Demia\'nski case, $\left.\psi\right|_q$ can be set to zero, which implies $\mathcal{E}=0$. Then, the derivatives $\left.\left(\partial^n_A\psi\right)\right|_q$ for $A=1,2$ will determine if the ring singularity is accessible. We obtain for the lowest order non-vanishing derivative the following expression

\begin{equation}
\left.\left(\partial_x^2\psi\right)\right|_q=\left.\left(\partial_y^2\psi\right)\right|_q=2(L^2\kappa-\mathcal{L}^2).
\label{dis2WH}
\end{equation}

This, yet again, is a negative quantity for time-like geodesics of arbitrary angular momentum $p_3=\mathcal{L}$. For null geodesics with zero angular momentum, (\ref{dis2WH}) vanishes. However, with said restriction, we have set all of the constants of motion to zero and hence reducing the motion of the particle to a trivial case. That is, the particle remains at a constant set of coordinates, including time itself. Motion such as this can be considered as unphysical behavior. Thus, we can conclude that no causal geodesic can reach the ring singularity $\sigma$ in this space-time.

At this point we should not forget about the remaining singularity $\sigma'$ which occurs when $\Delta_s=0$. This equation can be rearranged to a more familiar form $x^2+(y\pm a/2L^2)^2=a^2/4L^4$, i.e. $\sigma'$ is described by two circles in the $x-y$ plane. We can write the singular curvature set for this singularity as $\sigma'_X=\sigma'_+\cup\sigma'_-$, where

\begin{equation}
\sigma'_\pm=\left\{\left(t,v\cos\varphi,v\sin\varphi,u^3\right)\mid-\infty<t<\infty,\,0\leq\varphi<2\pi,\,x^2+(y\pm a/2L^2)^2=a^2/4L^4\right\},
\label{sigmapm}
\end{equation}

with $v=L\sqrt{(x^2+1)(1-y^2)}$ and $u^3=Lxy$. Interestingly enough, this singular curvature set contains that of the ring singularity such that, $\sigma_X=\sigma'_+\cap\sigma'_-$. By changing back to Boyer-Lindquist coordinates, and then to the Cartesian-like coordinates $\left\{u_1,u_2,u_3\right\}$, we can correctly visualize the shape of the singularity. Before doing so though, and since $y\in[-1,1]$, one can realize by examining (\ref{sigmapm}) that $a/L^2=1$ is a limiting case for the topology of $\sigma'_\pm$. For values $a/L^2>1$, we will obtain two closed line segments in the $x-y$ plane, rather than the previously described pair of circles that occur only when $a/L^2\leq1$. As a result, depending on the parameter $a/L^2$ the singular curvature set $\sigma'_X$ can have different geometrical properties. Furthermore, taking into account the azimuthal symmetry of the metric, we have that 

\begin{equation}
\sigma'_\pm\cong\begin{cases} 
      S^1\times S^1\times\mathbb{R} & \text{if } a/L^2\leq1, \\
      S^2\times\mathbb{R} & \text{if } a/L^2>1.
      \end{cases}
\label{sing}
\end{equation}

In the first case of (\ref{sing}), the ``spatial part'' of the singular curvature set $\sigma'_X$ is homeomorphic to two tori (one for each universe) which intersect at the ring singularity and at some other point of the throat. For the case $a/L^2>1$,  this singularity consists of two deformed two-spheres that completely surround the throat of the WH, making it impossible for any test particle to cross it. This is shown in Fig. \ref{fig:singularidad}.

\begin{figure}[htp]
	\centering
		\includegraphics[scale=1.1]{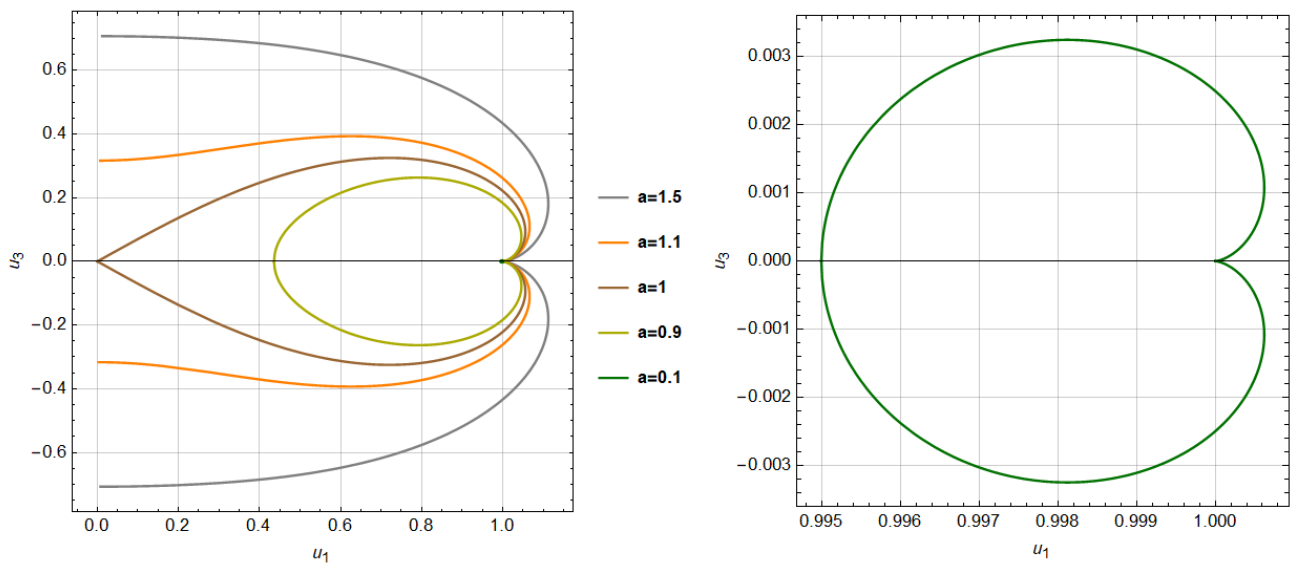} 
		\caption{Cross section of the singularity $\sigma'$ in the plane $u_1$-$u_3$ for different values of $a/L^2$, for simplicity we have made $L=1$. Here, $u_3$ is the symmetry axis. In the right panel we show a particular case where $a=0.1$, and hence $a/L^2\ll1$.}
	\label{fig:singularidad}
\end{figure}

Unfortunately, some complications arise from the structure of this singularity. Its topology is no longer $\Sigma^1\times\mathbb{R}$, instead it is $\Sigma^2\times\mathbb{R}$, i.e. a two-manifold $\times$ ``time". From the previous argumentation, we already know that there is $\sigma$-avoidance for causal geodesics, but what about the rest of the points of $\sigma'_X$? Note that the size of the singularity depends on the unit-less space-time parameter $a/L^2$. So, for $a/L^2\geq1$ there will always exist points $q'(x_s,y_s)\in\sigma'_X$ far from the ring singularity where $\left.\psi\right|_{q'}\geq0$. Therefore, by virtue of theorem 1, geodesic curves can meet the singularity $\sigma'$. Nevertheless, restricting the parameter to $a/L^2\ll1$ which corresponds to a slowly rotating wormhole, the region of the singularity $\sigma'$ is shrunk to a small neighborhood of $\sigma$ (see right panel of Fig. \ref{fig:singularidad}). A neighborhood of that sort has already been proven to be inaccessible to observers in geodesic motion. Additionally, it can be seen that the singularity $\sigma$ is the only possible source of affine incompleteness since, substituting the tensors (\ref{XYuv}) in (\ref{Xi}), the equations of motion show no ill-behavior for other points of the space-time\footnote{The divergence of the tensor $\mathcal{Y}^{\mu\nu}$ for $y=1$ is a consequence of the spheroidal coordinates here used and can be eliminated through a suitable change of coordinate system, e.g. the Cartesian-like coordinates $\left\{u_1,u_2,u_3\right\}$ mentioned previously in this section.}. Thus, and because the coordinate system $\left\{t,x,y,\varphi\right\}$ covers completely both universes, metric (\ref{ds2WH3}) describes a geodesically complete space-time, both for null and time-like curves, only for parameter values $a/L^2\ll1$. 

Despite the absence of an event horizon in this metric, the curvature singularities of the space-time cannot be observed by test particles in free-fall through the WH. As a passing note we point out that the Killing vector $X_0=\partial/\partial t$ becomes space-like inside the compact hyper-surface defined by the singularity $\sigma'$, while outside of it is time-like as expected in an asymptotically flat space-time.

\section{V. Conclusions}
We have formulated a series of criteria regarding causal geodesics and curvature singularities in stationary and axially symmetric space-times with a quadratic first integral. The criteria were stated in two theorems. The first one establishes the sufficient and necessary conditions for which time-like and null geodesics in such space-times can meet the singularity. The second one determines sufficient conditions for the existence of complete causal geodesics that encounter the singularity. Afterward, the theorems were applied to the asymptotically flat class of Pleba\'nski-Demia\'nski metrics which physically describe general black holes, and geometrically belong to the type D algebraic classification of space-times. It was found that in the electrically charged space-times of that class, the singularity is only reached by null geodesics with a specific relation of energy and angular momentum. These curves correspond to the principal null rays of the metric and do not become incomplete, at least, within a sufficiently small neighborhood of the curvature singularity. This last feature is shared with the Pleba\'nski-Demia\'nski space-times without electric charge too. On the other hand, time-like geodesics are repelled from the singularity in the electrically charged metrics, while in the uncharged ones there are incomplete time-like and null geodesics in contact with it. Finally, based on the derived theorems, we presented an example of a causal geodesically complete space-time that has a wormhole geometry with unbounded curvature. 

Unfortunately, as our results rely heavily on the separability of the equations of motion, a generalization to any axially symmetric line element seems unlikely through this approach. Naturally, these theorems can be also used for stationary space-times with spherical symmetry. \\

 \textbf{Acknowledgments.} This work was partially supported by CONACyT M\'exico under grants CB-2011 No. 166212, CB-2014-01 No. 240512, Project
No. 269652, Fronteras Project 281, and grant No. I0101/131/07 C-234/07 of the Instituto Avanzado de Cosmolog\'ia (IAC) collaboration (http://www.iac.edu.mx/). J.C.A. acknowledges financial support from CONACyT doctoral fellowships too. \\

\section{References}

\bibliography{bibliografia}

\end{document}